\documentclass[10pt,journal]{IEEEtran}
\usepackage[utf8]{inputenc} 
\usepackage{float}
\usepackage{enumerate}
\usepackage{balance}
\usepackage{dsfont}

\usepackage{color}
\usepackage{cite}
\usepackage{caption}
\usepackage[caption=false]{subfig}
\usepackage{algorithm}
\usepackage[noend]{algpseudocode}
\algrenewcommand\Return{\State \algorithmicreturn{} } 
\usepackage{amsmath}
\usepackage{amssymb}
\usepackage{amsthm}
\usepackage{array}
  \usepackage{graphicx}
\usepackage{epsf} 
\usepackage{psfrag}
\usepackage{fixltx2e}
\newtheorem{theorem1}{Theorem}
\newtheorem{proposition}[theorem1]{Proposition}
\newtheorem{theorem2}{Theorem}
\newtheorem{definition}[theorem2]{Definition}
\newtheorem{theorem3}{Theorem}
\newtheorem{corollary}[theorem3]{Corollary}

\newtheorem{theorem5}{Theorem}
\newtheorem{assumption}[theorem5]{Assumption}

\begin{document}
\title{ A Satisfactory Power Control for 5G Self-Organizing Networks}
\author{Hajar El Hammouti,~\IEEEmembership{Student Member,~IEEE,}
        Essaid Sabir,~\IEEEmembership{Senior Member,~IEEE,}
        and~Hamidou Tembine,~\IEEEmembership{Senior Member,~IEEE}
\thanks{Hajar El Hammouti is with the Department of Systems, Telecommunications, Networks, and Services (STRS), National Institute of Posts and Telecommunications, Rabat, Morocco. E-mail: elhammouti@inpt.ac.ma.}
\thanks{Essaid Sabir is with Networking Systems and Telecommunications (NEST) Research group, National School of Electricity and Mechanics, Casablanca, Morocco. E-mail: e.sabir@ensem.ac.ma.}
\thanks{Hamidou Tembine is with Learning \& Game Theory Lab, New York University Abu Dhabi, United Arab Emirates. E-mail:tembine@nyu.edu.}}


\maketitle

\begin{abstract}
Small-Cells are deployed in order to enhance the network performance by bringing the network closer to the user. However, as the number of low power nodes grows increasingly, the overall energy consumption of the Small-Cells base stations cannot be ignored. A relevant amount of energy could be saved through several techniques, especially power control mechanisms. In this paper, we are concerned with energy-aware self-organizing networks that guarantee a satisfactory performance. We consider satisfaction equilibria, mainly the efficient satisfaction equilibrium (ESE), to ensure a target quality of service (QoS) and save energy. First, we identify conditions of existence and uniqueness of ESE under a stationary channel assumption. We fully characterize the ESE and prove that, whenever it exists, it is a solution of a linear system. Moreover, we define satisfactory Pareto optimality and show that, at the ESE, no player can increase its QoS without degrading the overall performance. Under a fast fading channel assumption, as the robust satisfaction equilibrium solution is very restrictive, we propose an alternative solution namely the “long term satisfaction
equilibrium”, and describe how to reach this solution efficiently. Finally, in order to find satisfactory solution per all users, we propose fully distributed strategic learning schemes based on Banach-Picard, Mann and Bush-Mosteller algorithms, and show through simulations their qualitative properties.


\end{abstract}
\begin{IEEEkeywords}
Banach-Picard algorithm, Bush-Mosteller algorithm, efficient satisfaction equilibrium, expected robust game, game theory, long term satisfaction equilibrium, Mann iterates, satisfaction equilibrium, self-organizing networks, ultra-densification, 5G. 
\end{IEEEkeywords}
\IEEEpeerreviewmaketitle
\section{Introduction}
According to the Cisco visual network index (VNI) report~\cite{Cisco2016}, the monthly global mobile data traffic has reached 3.7 exabytes in 2015, and is expected to increase nearly eightfolds attaining 30.6 exabytes by 2020. In order to cope with this sheer volume of data traffic, a natural move to the next generation of wireless communication systems (5G) is needed. This is achieved by involving key technologies~\cite{Jeffrey2014,Talwar2014} including networks ultra-densification. The main driver behind ultra-densification is to substantially increase existing macro-cellular networks capacity~\cite{Osseiran2014,Kang2013,Bennis2016}. Ultimately, forecasts predict that by the time 5G comes to fruition, there will be more base stations (BSs) than mobile handsets~\cite{Jeffrey2013}.

The growing number of BSs, typically low power nodes (micro, pico, and femto base stations), also called Small-Cells~\cite{Whitepaper}, gives rise to many new challenges. Especially, those related to heterogeneity, optimization, and scaling. Clearly, as the number of heterogeneous Small-Cells grows tremendously, the amount of interferences will increase significantly~\cite{Dulaimi2015} leading to the so-called \textit{curse of dimensionality}~\cite{Curse}. The amount of data, that is exchanged between BSs in order to handle interferences, raises exponentially with the number of interfering BSs (dimensionality) resulting in a heavily loaded network (mainly with signaling messages).

Self-organizing networks (SON) are by far the most important approach to rise above dimensionality issues related to Small-Cells deployments~\cite{SON}. Not only do self-organizing networks enable to the network an automated resource management and improve BSs coordination, but SON can significantly reduce operators CAPEX (capital expenditures) and OPEX (operating expenses), especially by reducing human intervention and optimizing available resources.

However, as the number of low power nodes grows increasingly, the overall energy consumption of Small-Cells base stations cannot be ignored. A relevant amount of energy could be saved~\cite{Tang2015,Wu2015}. Mechanisms for energy-aware nodes are more than desirable. SON should implement energy saving techniques in order to prolong the lifetime of the batteries and increase the energy efficiency.




	%

In order to achieve energy efficiency, several techniques such as sleep mode optimizations~\cite{Ashraf2010}, power control mechanisms~\cite{Holtkamp2013}, and learning algorithms can be used. In this work, we are interested in achieving energy efficiency through a satisfaction mechanism. Practically, instead of achieving the best network performance by maximizing QoS, which is generally energy costly, nodes can only target satisfactory QoS levels, and hence, work efficiently. The choice of satisfactory QoS levels is merely endorsed by the two following reasons:

\begin{enumerate}
	\item First, as the number of base stations is highly increasing, probably overtaking the number of mobile devices in the next few years, the average number of devices per cell will decrease significantly. In such a context, optimization based on the full buffer traffic assumption, i.e. the user always asks for maximizing its QoS, cannot hold true. Unfortunately, this case can lead to resource consuming requests, while saving energy is possible without any significant deterioration of the user-perceived rate~\cite{Alvarez2016}.
	\item Second, many mobile applications (e.g. real-time services such as video conferencing and online gaming) require only fixed data rates in order to run properly. Assuming additional QoS demand will convey wasteful resources~\cite{Wang2007}. 
\end{enumerate}
 The main contribution of this paper is to present a novel approach while dealing with energy efficiency in SON. 
Our paper addresses the following questions: how to reach a target QoS while minimizing energy consumption? Given a realistic wireless framework, how to select the most efficient power allocation that meets with users expectations?

\subsection{Related work}
To answer these questions \textit{efficiently} and \textit{satisfactorily}, a game theoretical approach is adopted. Game theory provides powerful tools~\cite{Saad2016} and gives clear insights on interacting nodes behaviors. 

One of the well-known game theoretical solution concepts is the \textit{Nash equilibrium} (NE). NE is a strategy profile where no player has the incentive to deviate unilaterally. It has been shown that the NE generally fails to model the network performance. Indeed, when each player acts selfishly by increasing its power, subsequent interferences increase driving the network to a suboptimal situation.


Alternatively, in order to support the users QoS requirements, the  \textit{constrained Nash equilibrium} (CNE, also called \textit{generalized Nash equilibrium}) is introduced~\cite{Debreu1952,Menache2008,Sabir2009}. Particularly, constrained games are concerned with payoffs maximization (or minimization) subject to coupled and/or orthogonal constraints over the players strategies and/or payoffs. Hence, at a CNE, each player aims at achieving its optimal utility while satisfying QoS constraints. The CNE is designed to accommodate with QoS requirements that the NE fails to model. Nevertheless, from a practical point of view, the CNE can be a very restrictive solution that (i) reduces the set of players strategies, and (ii) requires costly efforts. More precisely, in order to reach the highest payoffs, greater efforts, such as higher powers, are generally needed. This may lead to a lower energy efficiency and cost effectiveness of the network. 

Consequently, a less restrictive solution concept, namely \textit{satisfaction equilibrium} (SE), has been introduced~\cite{Simon1956,Meriaux2012}. Mainly, in a less restrictive framework, players can only target satisfactory QoS levels without aiming at achieving the highest payoffs. At an SE, utilities optimization assumption is relaxed. The payoffs should only be above given thresholds based on users services requirements. Hence, (i) energy costs related to payoffs maximization are saved, and (ii) players constraints are satisfied. Note that the CNE can also be seen as an SE of a satisfaction game, since players constraints are always satisfied at a CNE~\cite{Perlaza2012}. Yet, the reverse is not necessarily true. 

The energy efficiency function was introduced by Meshkati \textit{et al.} in~\cite{Meshkati2005}. This function measures the performance of the network per Joule of consumed energy. For each node, the ratio: $\frac{\text{QoS}}{\text{Consumed energy}}$ is reduced. The work in~\cite{Saad2016Bis} investigates energy efficiency in ultra-dense networks through joint power control and users scheduling. The problem of energy efficiency maximization is formulated as a dynamic stochastic game and cast as a mean-field game. The authors show that the mean-field equilibrium saves energy and reduces outage probability. In~\cite{MIMO}, the energy efficiency for multiple input multiple output antennas heterogeneous networks is studied using a non-cooperative and cooperative power control game. The authors propose a power allocation algorithm in order to reach Pareto optimal solutions.


Authors in~\cite{TangJSAC} transform the energy efficiency problem maximization to a power minimization problem. A power allocation algorithm based on semidefinite programming is therefore proposed to achieve energy efficiency. In~\cite{Debbah}, Lakshminarayana \textit{et al.} are also interested in power minimization. They present the problem of saving energy in Small-Cells under time constraints. The impact of delayed information is analyzed and a Lyapanov optimization is proposed.

Although game theory has been widely adopted to model network dynamics, very little attention has been paid to satisfactory solutions.
 In~\cite{Perlaza2012}, Perlaza \textit{et al.} investigate conditions for existence of satisfaction equilibria in a general framework of
QoS provisioning in self-configuring networks. Authors also study conditions for existence and uniqueness of efficient SE (ESE). Sufficient conditions for convergence of distributed algorithms that reach SE are presented by Meriaux \textit{et al.} in~\cite{Meriaux2012}. Authors in~\cite{Perlaza2012Bis} present convergence analysis to an SE, ESE, and satisfactory solutions that are not equilibria. 



\subsection{Contribution}
In this work, we address the problem of energy-aware user satisfaction in SON. Our main objective is to meet the users requirements while reducing the energy consumption. Accordingly, we aim at seeking for satisfaction equilibria, mainly the ESE in both stationary and random channels context.

For this paper, while we refer to earlier works~\cite{Perlaza2012} and~\cite{Perlaza2012Bis}, the focus is different:
\begin{enumerate}
	\item First, when the channels follow a slow fading model:
\begin{itemize}
	\item we fully characterize the ESE while considering information-theoretic transmission rate based measure, and prove that, whenever it exists, it is the unique solution of a given linear system.
\item We prove that satisfactory QoS levels and channel states are correlated to the transmit powers. For a desired QoS level, the users need to adjust their powers according to their channel states. Specifically, when the channel state decreases, the required transmit power increases.
\item We introduce ``satisfactory Pareto optimality'' property and show that the ESE is satisfactory Pareto optimal. Mainly, at the ESE, no user can increase its payoff without dissatisfying its opponents.
 
\end{itemize}
\item Second, when the channels are fast varying over time:
\begin{itemize}
	\item we define the ``robust satisfaction equilibrium'' and introduce a novel solution concept, namely the “long term satisfaction equilibrium”, which is less restrictive, and hence, more energy efficient.
	\item In order to achieve energy efficiency, we seek for power allocations that minimize the channels variance while averagely respecting users requirements. 
	
\end{itemize}
 \item One of the main contributions of this paper is that we propose fully distributed algorithms: (i) based on Banach-Picard iterations in order to reach the ESE for continuous power states and stationary channels, (ii) a modified version of Banach-Picard referred to as ``the progressive Banach-Picard algorithm for capacity discovery'' to reach the maximum network capacity, (iii) an updated version of Bush-Mosteller algorithm in order to converge to the ESE when the channels are stationary and the power states spaces are discrete, (iv) and Mann iterates based algorithm for fast varying channels in order to reach the efficient long term satisfaction
equilibrium.
\item Finally, we show the qualitative properties of the proposed algorithms through simulation results. 

\end{enumerate}

\subsection{Structure}
The remainder of the paper is organized as follows. The next section presents the system model. Section~\ref{Stationary} formulates the problem as a satisfaction game under a stationary channel assumption and describes some properties related to the set of satisfacion equilibria. A full characterization of ESE is also provided and the Pareto optimality property is studied. In section~\ref{Random}, under channel randomness assumption, satisfactory solutions are investigated. In section \ref{Fully}, we present distributed schemes in order to reach efficient solutions for both stationary and random channels. Section \ref{Simulations} presents numerical results to illustrate the performance of the studied algorithms. Finally, in the last section, we make a few concluding remarks. The proof of proposition \ref{propo1} is given in Appendix and important notations are summarized in Table~\ref{tab1}.
\section{System Model }\label{SystemModel}
Consider a self-organizing network where $N$ heterogeneous mobile stations (MSs) communicate with a common concentration point (e.g. a base station). We denote by $\mathcal{U}$ the set of (MSs). Users operate over the same radio channel. The radio channel is first supposed stationary, and, then, fast varying within time. The channel is considered slow or fast varying depending on the time the channel states change compared to the coherence time. We denote by $h_i$ the power fading gain between user $i$ and the concentration point. Each user selects its transmit power within a bounded power space that we denote by $\mathcal{P}^i$. 
Each user $i$ aims at guaranteeing a target throughput $\theta_i$. We denote by $r_{i}$ the bandwidth-normalized instantaneous throughput of node $i$, 
\begin{equation}\label{th}
r_i(\textbf{P})=\log_2(1+\gamma_i(\textbf{P})),
\end{equation}
	\noindent where $\gamma_{i}$ denotes the instantaneous signal-to-interference-and-noise-ratio (SINR) of user $i$,
	\begin{equation}
	\gamma_{i}(\textbf{P})=\frac{h_{{i}}P_i}{\eta+\sum \limits_{j \in \mathcal{U} \backslash\{i\}} h_{{j}}P_j},
	\end{equation}
 \noindent	with $\textbf{P}$ the power vector that encompasses all MSs transmit powers. $P_i$ is the power selected by user $i$, and $\eta$ is the variance of a Gaussian random variable that represents the additive Gaussian noise. 

Under a fast fading channel assumption, only expected throughput can be considered
\begin{equation}\label{Expected}
\bar{r}_{{i}}(P)\quad=\mathbb{E}_{h_1,\dots,h_{{N}}}\left[\log(1+\gamma_{i}(P))\right],
\end{equation}
\noindent  where~$\mathbb{E}_{h_1,\dots,h_{{N}}}[.]$ denotes the mathematical expectation with respect to the joint random variable~${h_1,\dots,h_{{N}}}$.

The main notations of this work are summarized in Table~\ref{tab1}.
\begin{table}[t]
\begin{tabular}{l l}
\hline
Symbol & Meaning\\
\hline
$\mathcal{U}$& set of MSs\\
$N$& number of MSs\\
$\mathcal{P}^i$& power states space related to user $i$\\
$\textbf{P}$ & power vector, it belongs to $\prod \limits_{i\in \mathcal{U}}\mathcal{P}^i$ \\
$P_i$ & power of MS $i$\\
$\gamma_i$ & instantaneous SINR related to MS $i$ \\
$\eta$ & additive Gaussian noise \\
$h_{{i}}$ & power fading gain between the concentration point and \\ & user~${i}$\\
$r_{{i}}$ & instantaneous throughput related to user~${i}$\\
$\bar{r}_{{i}}$ & expected throughput related to user~${i}$\\
$\theta_i$ & target throughput related to user $i$\\
$P^t$& $=\sum\limits_{j\in \mathcal{U}}P_jh_j$ the total received power at the \\& concentration point\\
$P^{max}_i$ & maximum power allowed to user $i$\\ 
$\mathbb{E}_{h_1,\dots,h_{{N}}}[.]$ & the mathematical expectation with respect to the joint\\ & random variable~${h_1,\dots,h_{{N}}}$.\\
$C$& maximum network capacity\\
$\mathcal{S}$& set of ESE\\
\hline
\end{tabular}
\caption{Summary of notations}
\label{tab1}
\end{table}

\section{Stationary Channel}\label{Stationary}

Here, we suppose a slow fading channel. Specifically, the channel remains constant during the coherence interval. Hence, only instantaneous throughput.
\subsection{Game formulation}

Suppose the non-cooperative game $\mathcal{G}~=~\{\mathcal{U},\{\mathcal{P}^i\}_{i \in\mathcal{U}},\{r_i\}_{i \in\mathcal{U}},\{\theta_i\}_{i\in \mathcal{U}}\}$ where $\mathcal{U}$ is the set of players, $\mathcal{P}^1,\dots, \mathcal{P}^N$ are the sets of pure strategies of each user.The payoff of each user $i$ $r_i$ is given by its throughput. We will separate the strategy of a node $i$, $P_i$, and its opponents by using the following notation $(P_i,\textbf{P}_{-i})$. Each user $i$ aims at achieving a minimum throughput $\theta_i$. Consequently, each user $i$ is satisfied with its strategy $P_i$ given the others strategies $\textbf{P}_{-i}$, if 
\begin{equation}
r_i(P_i,\textbf{P}_{-i})\geq \theta_i.
\end{equation}

\subsection{Satisfaction equilibrium}
In order to achieve satisfactory solutions, only a set of strategies is allowed for each user given its opponents strategies. We refer to these strategies as ``\textit{feasible strategies}''. Such a set can be characterized by a correspondence function $f_i:\mathcal{P}_{-i}\longrightarrow 2^{\mathcal{P}^i}$, where $\mathcal{P}_{-i}=\mathcal{P}^i\times \dots \times \mathcal{P}^{i-1} \times \mathcal{P}^{i+1}\times \dots \times \mathcal{P}^N$, and $2^{\mathcal{P}^i}$ is the set of all subsets of $\mathcal{P}^i$. For each $\textbf{P}_{-i}\in \mathcal{P}_{-i}$,  if $ P_i \in f_i(\textbf{P}_{-i})$, then $r_i(P_i,\textbf{P}_{-i})\geq \theta_i$. Accordingly, we define the set of satisfaction equilibria as follows. 

\begin{definition}[Satisfaction equilibrium]
A strategy profile $\textbf{P}^+$ is a \textit{satisfaction equilibrium} of a game $\mathcal{G}=\{\mathcal{U},\{\mathcal{P}^i\}_{i \in\mathcal{U}},\{r_i\}_{i \in\mathcal{U}},\{\theta_i\}_{i\in \mathcal{U}}\}$, if 
\begin{equation}\label{sat}
\forall i \in \mathcal{U}, \text{ } \forall P_i^+ \in f_i(\textbf{P}_{-i}^+),\text{ i.e., } r_i(P_i^+,\textbf{P}_{-i}^+)\geq \theta_i.
\end{equation}

\end{definition}

It is important to note that an SE exists, if there exists $\textbf{P}^+$ that satisfies constraints~(\ref{sat}), but also if the constraints are feasible, mainly, if  
\begin{equation}\label{Unfeasible}
\sum \limits_{i\in \mathcal{U}} \theta_i \leq \mathcal{C},
\end{equation}
where $\mathcal{C}$ is the maximum network capacity (which is also solution to the rate-maximization problem as described~\cite{Chee2013}). Hereafter, we assume the constraints feasibility and discuss SE existence. The impact of constraints feasibility will be discussed in section~\ref{Simulations}.
\begin{proposition}\label{Convex}
For each $i\in \mathcal{U}$, let $\mathcal{P}^i=[0,P_i^{max}]$ be the set of pure strategies of player $i$, $P_i^{max}$ is the maximum transmit power value allowed to $i$.

\noindent Then, the set of SE of game $\mathcal{G}$ is convex, closed and bounded.
\end{proposition}
\begin{proof}
\begin{itemize}
\item[]
	\item	Convex: SE are solution to linear inequalities system, each feasible region (inequality) is convex~\cite{solo1980}, and the intersection of convex sets is convex.
\item	Closed: the closeness of the SE set stems from the soft inequalities. Each inequality is a union of two closed half planes, and the union and intersection of a finite number of closed sets is a closed set.
\item	Bounded: this property holds true because the set of power states is supposed bounded.
\end{itemize}
\end{proof}

\subsection{Efficient satisfaction equilibrium}\label{ESE}
We denote by $\mathcal{S}$ the set of SE. Whenever $\mathcal{S}$ is non-empty, we define the \textsl{efficient satisfaction equilibrium} (ESE) concept as provided in~\cite{Perlaza2012}.
\begin{definition}[Efficient satisfaction equilibrium]
An SE $\textbf{P}^+$ is said efficient if it minimizes the efficiency function $F(\textbf{P})=\sum \limits_{i\in \mathcal{U}} P_i$ with $\textbf{P} \in \mathcal{S}$.
\end{definition}
The next corollary prove the existence of an ESE when the set of SE is non-empty. 
\begin{corollary}
If the set of SE of the game $\mathcal{G}$ is non empty, an efficient satisfaction equilibrium exists.   
\end{corollary}
\begin{proof}
The efficiency function $F(\textbf{P})=\sum \limits_{i\in \mathcal{U}} P_i$ with $\textbf{P} \in \mathcal{S}$ is convex with respect to $\textbf{P}$ over the set of SE which is convex and compact (Proposition~\ref{Convex}). Hence a minimum of $F$ exists.
\end{proof}

\begin{proposition}[ESE uniqueness and characterization]\label{uni}
Let $\mathcal{P}^i=[0,P_i^{max}]$ be the set of pure strategies of player $i$. Suppose the set of SE is non-empty. The ESE is unique and it is the solution to the linear system $\forall i \in \mathcal{U}$, $r_i(\textbf{P})=\theta_i$. 
\end{proposition}
\begin{proof}
Let $\textbf{P}^+$ be the solution to the linear system, such that $\forall i \in \mathcal{U}, r_i(\textbf{P}^+)=\theta_i$. We want to prove that $\textbf{P}^+$ minimizes the efficiency function $F(\textbf{P})=\sum \limits_{i\in \mathcal{U}} P_i$ with $\textbf{P} \in \mathcal{S}$.

\noindent Suppose by absurd that it exists $\tilde{\textbf{P}}\! \in \! \mathcal{S}$ such that \begin{equation}\sum \limits_{i\in \mathcal{U}}\!\! \tilde{P}_i\!~<~\!\sum \limits_{i\in \mathcal{U}}\!\! P^+_i.\end{equation} 

\noindent We want to prove that under this assumption, there is at least one player whose target throughput is not satisfied. Suppose that we have for each $i$ such that $1\!\leq\! i\leq \!N\!-\!1$

\begin{equation}\label{vrai}
r_i(\tilde{\textbf{P}})\geq \theta_i.
\end{equation}
In order to show that the requirement of the player $N$ is unsatisfied, meaning $r_N(\tilde{\textbf{P}})\!<\!\theta_N$, it is sufficient to prove that
\begin{equation}\label{Absurd}
\log_2(\!1+\!\frac{\tilde{P}_Nh_N}{\eta+\tilde{P}^t-\tilde{P}_kh_k}\!)\!<\!\log_2(\!1+\!\frac{\tilde{P}_Nh_N}{\eta+{P}^{t+}-{P}^+_Nh_N}\!), 
\end{equation}
with ${P}^t=\sum \limits_{i\in \mathcal{U}} P_ih_i$ the overall received power at the concentration point. 
In order to prove equation (\ref{Absurd}), it is sufficient to prove that 
\begin{equation}\label{Absurd2}
\frac{\tilde{P}_Nh_N}{\eta+\tilde{P}^t-\tilde{P}_Nh_N}<\frac{\tilde{P}_Nh_N}{\eta+{P}^{t+}-{P}^+_Nh_N}. 
\end{equation}
Or alternatively, it is sufficient to prove that
\begin{equation}\label{eq}
\eta(\tilde{P}_N-P_N^+)+\tilde{P}_NP^{t+}-P^+_N\tilde{P}^t<0.
\end{equation}
Note that equation~(\ref{eq}) is obtained after a few algebras from equation~(\ref{Absurd2}). 
\noindent Important to note that $\forall i, r_i(\textbf{P}^+)=\theta_i$. Hence, using the same previous reasoning, equation~(\ref{vrai}) becomes, $\forall 1\leq i\! \leq\! N-1$
\begin{equation}\label{vrai2}
\eta(\tilde{P}_i-P_i^+)+\tilde{P}_iP^{t+}-P^+_i\tilde{P}^t\geq 0.
\end{equation} 

\noindent By subsequent summation of inequalities described by equation~(\ref{vrai2}) from $i\! =\!1$ to $i\! =\!N\!-\!1$, we obtain

\begin{equation}
\eta(\tilde{P}^t\!\!-\! \tilde{P}_N\! h_N\!\!-\! P^{t+}\!+\! P_N^+\! h_N\! )\! +\! P^{t+}\! (\! \tilde{P}^t\!\!-\!\tilde{P}_N\!h_N\! )\!-\! \tilde{P}^t\! \!(\! P^{t+}\! -\! P_N^+\! h_N\! )\!\!\geq\!\! 0.
\end{equation}
After a few algebras, we obtain
\begin{equation}
\eta(\tilde{P}_N-P_N^+)+\tilde{P}_NP^{t+}-P^+_N\tilde{P}^t< \frac{\eta}{h_k}(\tilde{P}^t-P^{t+}).
\end{equation}

Since we supposed that $\sum \limits_{i\in \mathcal{U}} \tilde{P}_i < \sum \limits_{i\in \mathcal{U}} P^+_i$, then $\tilde{P}^t<P^{t+}$ (the total received power conserves the order as the channels gain are stationary and the same for both received powers).
Thus, 
\begin{equation}
\eta(\tilde{P}_N-P_N^+)+\tilde{P}_NP^{t+}-P^+_N\tilde{P}^t<0
\end{equation}
which is equivalent to say $r_N(\tilde{\textbf{P}})<\theta_N$, and this is a contradiction since $\tilde{\textbf{P}}$ is an SE. Thus, $\sum \limits_{i\in \mathcal{U}} \tilde{P}_i > \sum \limits_{i\in \mathcal{U}} P^+_i$,
and $\textbf{P}^+$ is the ESE.
\end{proof}
Proposition~\ref{uni} states that, whenever it exists, the ESE is met when all players are barely enough satisfied. The obtained result confirms intuitive expectations since only minimal efforts are needed to afford minimal satisfaction. Furthermore, the Pareto frontier is exactly defined by the efficient satisfaction equilibrium. The next definition describes the Pareto efficiency in satisfaction games. 
  
\begin{definition}[Satisfactory Pareto optimality]
In a satisfaction game, a strategy profile is Pareto optimal if no player can increase its payoff without dissatisfying its opponents.  
\end{definition}
The following corollary arises from the definition above.
\begin{corollary}[Pareto optimality]
Suppose the set of SE is non-empty and the power spaces are non-empty, convex and compact sets. The ESE exists and is satisfactory Pareto optimal. 
\end{corollary}
\begin{proof}
At the ESE, given the other players strategies, if a player $i$ decides to increase its payoff by increasing its transmit power, the interferences increase, and thus all the other players payoffs get deteriorated, mainly, they will be under users requirements thresholds. 
\end{proof}
In order to reach their QoS requirements, users should adapt their transmission powers according to their channel states. We show, through the following proposition, the impact of the channel state and the users demands on the powers decision-making. 
\begin{proposition}\label{order}
Let $\textbf{P}^+$ be the ESE of the game $\mathcal{G}$. If, for users $i$ and $j$, $\theta_i \leq \theta_j$, then $P_i^+h_i \leq P_j^+h_j$.
\end{proposition}
\begin{proof}
The proof stems from the increasing nature of the payoff functions. Indeed, for a given fixed power allocation of $i$'s opponents, $r_i(P_i,\textbf{P}_{-i})$ is increasing with respect to $P_i$.

First, at the ESE, $r_i(P_i^+,\textbf{P}_{-i}^+)= \theta_i$.

\noindent Suppose $\theta_i\leq\theta_j$, therefore, for $x \in \mathcal{P}^i $, we have 
\begin{equation}
g(x)=\log_2\left(1+\frac{x}{\eta+P^t-x}\right),
\end{equation}  
with $P^t=\sum\limits_{j\in \mathcal{U}}P_jh_j$. By deriving $g(.)$ with respect to $x$, we obtain
\begin{equation}
\frac{\partial g(x)}{\partial x}=\frac{\eta+P^t+2x}{\eta+P_t-x}\geq 0.
\end{equation}
Hence, when
\begin{equation}
\log_2\left(1+\frac{P_i^+h_i}{\eta+P^t-P_i^+h_i}\right)\leq \log_2\left(1+\frac{P_j^+h_j}{\eta+P_t-P_j^+h_j}\right), 
\end{equation}
we have $P_i^+h_i\leq P_j^+h_j$, which completes our proof.
\end{proof}

Next, we provide an analytical expression of ESE under existence conditions.
\begin{proposition}[Full characterization of ESE]\label{propo1}
Let $\mathcal{P}^i=[0,P_i^{max}]$ be the set of pure strategies of player $i$. The ESE of the game $\mathcal{G}$ exists if
 %
%

\begin{equation}\label{deterr}\Bigg\{
\begin{aligned}
\sum \limits_{i=1}^{N-1}\! h_i(1-2^{\theta_i})\!\!\prod \limits_{j=1,j\neq i}^{N} 2^{\theta_j}h_j\!\!+\!\! h_N\!\! \prod \limits_{j=1}^{N-1}2^{\theta_j}h_j\neq 0,&\\
0\leq P_N^+\leq P_N^{max} \text{ and } P_i^+\leq P_i^{max}, 1\leq i\leq N-1,&
\end{aligned}
\end{equation}

such that
\begin{equation}\label{soll}
\Bigg\{
\begin{aligned}
P_N^+ &=\frac{\eta(2^{\theta_N}-1)}{\sum \limits_{i=1}^{N-1}h_i\frac{1-2^{\theta_N}}{2^{\theta_N}h_N}\frac{2^{\theta_i}h_i}{1-2^{\theta_i}}+\frac{h_N}{1-2^{\theta_N}}}\\ 
P_i^+ &=\frac{1-2^{\theta_i}}{2^{\theta_i}h_i}\frac{2^{\theta_N}h_N}{1-2^{\theta_N}}P_N, 1\leq i\leq N-1
\end{aligned}
\end{equation}
Specifically, $P_N^+$ and $P_i^+,1\leq i\leq N-1$ are the players strategies at the ESE.
\end{proposition}
%

\begin{proof}
Due to its length, the proof is moved to the Appendix~\ref{Appen}.
\end{proof}
\section{Random Channel}\label{Random}

In real wireless communication systems, the channel state information (CSI) is aquired through training sequencies (pilots). The CSI is estimated at the receiver and fed back to the transmitter. Even though in time division duplex (TDD) systems, the CSI can be estimated instantaneously at the transmitter using the reverse channel, the CSI can be noisy, and thus, it is best estimated statistically. Conversely, in frequency division duplex (FDD) systems, the channel can only be estimated statistically.

Furthermore, the assumption of perfect CSI is not grounded in reality especially for fast fading channels where the channels vary faster than the coherence time, and delay should be estimated.

Under realistic CSI assumptions, only expected throughput can be used as described by equation~(\ref{Expected}). 

Suppose $\mathcal{G}'=\{\mathcal{U},\{\mathcal{P}^i\}_{i \in\mathcal{U}},\{\bar{r}_i\}_{i \in\mathcal{U}},\{\theta_i\}_{i\in \mathcal{U}}, \mathbb{E}_{h_1,\dots,h_{{N}}}[.]\}$, an expected robust game with $\{\bar{r}_i\}_{i \in\mathcal{U}}$ as the players expected utilities. The expected values are over the joint random variables~${h_1,\dots,h_{{N}}}$ which are supposed independent and identically distributed (i.i.d.). We assume that all the channels are bounded and different from $0$, this is a practical assumption endorsed by the fact that radio channels take strict positive values. Hence, we have the following assumption:

\begin{assumption}\label{assump}
$\forall i \in \mathcal{U}, 0<h_i\leq h^{max}.$
\end{assumption}

Under Assumption~\ref{assump}, we define the \textit{robust satisfaction equilibrium} (RSE).

\begin{definition}[Robust satisfaction equilibrium]
$\textbf{P}^*$ is called a robust satisfaction equilibrium iff $\forall 0<h_i\leq h^{max}, \text{ such that } i \in \mathcal{U}, r_i(\textbf{P}^*)\geq \theta_i$.
\end{definition}

Hence, the RSE is a power allocation that enables to the users to be always satisfied regardless the state of their respective channels. The RSE is particularly distribution-free since it does not depend on the channels states. It is worth noting that when Assumption~\ref{assump} is not satisfied, mainly, if $\exists i \in \mathcal{U},\text{ }h_i=0$, all the strategies of player $i$ will be unfeasible, and hence the non-existence of an RSE.

Moreover, the existence of an RSE is a stringent condition as it is very restrictive to find a power allocation that satisfies all the players for all their channel states. At least, this power allocation could not be efficient. Thus, we define a less restrictive solution, namely the ``\textit{long term satisfaction equilibrium}'' (LTSE).

\begin{definition}[Long term satisfaction equilibrium]
$\textbf{P}^+$ is called a long term satisfaction equilibrium if it is a solution to the following problem
\begin{equation}\label{sollli}
\Bigg\{
\begin{aligned}
\min_P \mathbb{E}_{h_1,\dots,h_{{N}}}\left[(r_i(P)-\bar{r}_i(P))^2\right] ,&\\
\text{subject to } \forall \text{ } i\in \mathcal{U}, \bar{r}_i(P)\geq \theta_i.&
\end{aligned}
\end{equation}
 
\end{definition}

Indeed, in order to achieve satisfaction, we seek for power allocations that minimize the channels variance while \textit{averagely} respecting users requirements. This assumption is less restrictive than instantaneous satisfaction where users should be satisfied all the time. 

Reducing variance $Var(r_i(P))=\mathbb{E}_{h_1,\dots,h_{{N}}}\left[(r_i(P)-\bar{r}_i(P))^2\right]$ ensures small instantaneous throughput fluctuations around the expected throughput. Furthermore, in order to ensure energy efficiency, the expected throughput can be replaced by the users requirements. This is important to avoid energy wasting solutions that guarantee to the users higher payoffs. Hence, the efficient LTSE is given by the solution of the following problem:



\begin{equation}\label{sollli}
\Bigg\{
\begin{aligned}
\min_P \mathbb{E}_{h_1,\dots,h_{{N}}}\left[(r_i(P)-\theta_i)^2\right],&\\
\text{subject to } \forall \text{ } i\in \mathcal{U}, \bar{r}_i(P)= \theta_i.&
\end{aligned}
\end{equation}

Another way to consider energy efficiency in a fully distributed fashion is to solve the following constrained problem:

\begin{equation}\label{sollli2}
\Bigg\{
\begin{aligned}
\text{minimize       } P_i, &\\
\text{subject to } \forall \text{ } i\in \mathcal{U}, \bar{r}_i(P)= \theta_i.&
\end{aligned}
\end{equation}

\begin{proposition}
Let $\textbf{P}^+$ be the power allocation that minimizes the transmit power of each user and exactly guarantees their satisfaction (solution to the problem~(\ref{sollli2})), $\textbf{P}^+$ also ensures small variations around satisfactory levels (solution to the problem~(\ref{sollli})). 
\end{proposition}
\begin{proof}
$\log_2(.)$ is a Lipschizian function with respect to $P_i$, therefore, $\exists 0<k<1, c\in \mathds{R}^+ $, $\forall i \in \mathcal{U}$
\begin{equation}
\left|r_i(P_i,\textbf{P}_{-i})-\theta_i\right|^2<(k P_i+c)^2.
\end{equation}
Consequently, $\exists 0<k'<1, c'\in \mathds{R}^+ $
\begin{equation}
\mathbb{E}_{h_1,\dots,h_{N}}\left[(r_i(\textbf{P})-\theta_i)^2\right]<(k' P_i+c')^2,
\end{equation}

Accordingly, when the transmit powers are reduced, $\mathbb{E}_{h_1,\dots,h_{N}}\left[(r_i(\textbf{P})-\theta_i)^2\right]$ is reduced as well, which completes our proof.
\end{proof}

Next, we define the probability of meeting a satisfactory equilibrium $\mathbb{P}(\exists P| \textbf{r}(\textbf{P})\geq \theta)$ as described in the following equation

\begin{equation}
\mathbb{P}(\exists P| \textbf{r}(\textbf{P})\geq \theta)=\mathbb{P}\{\exists P\in \prod\limits_{i \in \mathcal{U}}\mathcal{P}^i,\forall \text{ } i\in \mathcal{U}, r_i(\textbf{P})\geq \theta_i\}. 
\end{equation}

In a random context, in order to meet satisfactory equilibria very frequently, the probability $\mathbb{P}(\exists P| \textbf{r}(\textbf{P})\geq \theta)$ is to be maximized.

%

\begin{proposition}
Let $\mathbb{P}(\exists P|r_i(P)\geq \theta_i)$ be the probability to meet an LTSE, therefore

\begin{equation}
\mathbb{P}(\exists P|r_i(P)\geq \theta_i)\leq \min (1,\frac{Var(r_i(P))}{(\theta_i-\bar{r}_i(P))^2}),
\end{equation}
\end{proposition}
\begin{proof}
According to the extended version of the Chebyshev's inequality for monotonically increasing functions,

\begin{equation}\label{Cheb}
\mathbb{P}(\exists P|r_i(P)\geq \theta_i)\leq \frac{\mathbb{E}_{h_1,\dots,h_{{N}}}(\Psi(r_i(P)))}{\Psi(\theta_i)},
\end{equation}
with 
\begin{equation}\label{var}
\Psi(r_i(P))=\left(r_i(P)-\bar{r}_i(P)\right)^2.
\end{equation}
By replacing equation~(\ref{var}) in equation~(\ref{Cheb}), 

\begin{equation}\label{Bounded}
\mathbb{P}(\exists P|r_i(P)\geq \theta_i)\leq \frac{Var(r_i(P))}{(\theta_i-\bar{r}_i(P))^2}.
\end{equation}
Equation~(\ref{Bounded}) shows that the probability of meeting an SE is bounded by $\frac{Var(r_i(P))}{(\theta_i-\bar{r}_i(P))^2}$, mainly, by $\min_P  \frac{Var(r_i(P)) }{(\theta_i-\bar{r}_i(P))^2}$.

Hence, 
\begin{equation}
\mathbb{P}(\exists P|r_i(P)\geq \theta_i)\leq \min (1,\frac{Var(r_i(P))}{(\theta_i-\bar{r}_i(P))^2}),
\end{equation}  
which completes our proof.
\end{proof}

%


\section{Fully Distributed Satisfaction}\label{Fully}
In this section, we propose distributed schemes in order to reach the satisfaction efficiently for different cases: stationary channels and continuous power space, capacity discovery for stationary channels and continuous power space, stationary channels and discrete power spaces, and random channels and continuous power spaces.

%

\subsection{Learning for stationary channels}
\subsubsection{Banach-Picard algorithm}
In order to reach the ESE under the stationary channels assumption and continuous power spaces, we propose a distributed learning scheme described by the Banach-Picard algorithm, also called ``fixed point iterations''. The Banach-Picard algorithm is known to be convergent with a geometrical rate to the solution of the equation $w(x)=x$~\cite{TembineBook}. In order to converge to the unique fixed point of $w(.)$, function should be contractive. In our case, we suppose $w(x)=\frac{x}{\theta_i}r_i(x,\textbf{P}_{-i})$, which is a contractive function. Notice that when the fixed point is reached, mainly $\frac{x}{\theta_i}r_i(x,\textbf{P}_{-i})=x$, we obtain, $r_i(x,\textbf{P}_{-i})=\theta_i$ (we suppose that the fixed point is different from zero otherwise nodes cannot target their satisfactory QoS levels).

The Banach-Picard algorithm is described in Algorithm~\ref{alg:Picard}.
\begin{algorithm}[H]
  \caption{Banach-Picard algorithm
    }\label{alg:Picard}
     \begin{algorithmic}[1] 
					\State{\textbf{Parameters initialization}}
					\Statex{Each user picks randomly a transmit power with initial probabilities vector $P_{i}^{0}$}
					\Repeat
					\State{\textbf{Learning pattern}}
      \For{Each user $i$}
			   \Statex{Observe the value of its instantaneous throughput ${r_i^t(P_i^t,\textbf{P}_{-i})}$ }
			   \Statex{Update its power as follows:}
			   \Statex{$P_{i}^{t+1} \gets  P_{i}^t\frac{\theta_{i}}{r_i^t}$ }
			\EndFor
					\Until{The stopping criterion}
  \end{algorithmic}
\end{algorithm}
$P_{i}^{t}$ is the power of user $i$ at instant $t$, and ${r_i^t(P_i^t,\textbf{P}_{-i})}$  is the throughput of player $i$ at instant $t$ when its opponents choose $\textbf{P}_{-i}$. The stopping criterion can be formulated as ``all players have reached their target throughput''. It can also be described by $P_i^t$ convergence rate: when $\left|P_i^{t+1}-P_i^t\right|<\rho$, with $\rho$ goes to zero, the algorithm stops.

It is important to note that Banach-Picard algorithm, and all the algorithms we propose in this section, are fully distributed since no information about the other players strategies, their throughput, or channels states is needed. Players need only to observe their own payoffs in order to pick up the best strategy efficiently.  

\subsubsection{Progressive Banach-Picard algorithm for capacity discovery}
When achieving the best performance is more important than energy consumption, an adapted version of Banach-Picard algorithm, referred to as the \textit{progressive Banach-Picard algorithm for capacity discovery}, can be used. Indeed, in the capacity discovery algorithm, once a user's request (demand) is reached, the user increases its demand slightly, and adjusts its power accordingly in order to reach its new target. This process allows to the users to maximize their payoffs reaching a constrained Nash equilibrium (maximizing their payoffs subject to their initial requirements constraints) and avoiding any under-utilization of the network. The progressive Banach-Picard algorithm for capacity discovery is described in Algorithm~\ref{alg:ban}.

\begin{algorithm}[H]
  \caption{Progressive Banach-Picard algorithm for capacity discovery
    \label{alg:ban}}
     \begin{algorithmic}[1] 
					\State{\textbf{Parameters initialization}}
					\Statex{Each user picks randomly a transmit power with initial probabilities vector $P_{i}^{0}$}
					\Repeat
					\State{\textbf{Learning pattern}}
      \For{Each user $i$}
			\Statex{Observe the value of its instantaneous throughput ${r_i^t(P_i^t,\textbf{P}_{-i})}$ }
			     \Statex{$P_{i}^{t+1} \gets  P_{i}^t\frac{\theta_{i}}{r_i^t}$ }
					\If{$\theta_i$ is reached and $\sum \limits_{i}\theta_i<C$}
					\Statex{Update its demand as follows:}
					\Statex{$\theta_{i}\gets\theta_i+\epsilon$}
					\EndIf
				   \EndFor
					\Until{The stopping criterion}
  \end{algorithmic}
\end{algorithm}
where the parameter $\epsilon \in [0,1]$ is the throughput step size.

%
%
%
\subsubsection{Adapted Bush-Monsteller algorithm}

In order to reach an SE in a distributed fashion under stationary channels assumption and discrete power spaces, we propose a modified version of the Bush-Monsteller Algorithm~\cite{Bush,scha} referred to as ``\textit{Adapted Bush-Monster algorithm}''. The Bush-Monsteller Algorithm is a well-known algorithm that converges to an NE. The modified version we propose reaches to the ESE of game $\mathcal{G}$. 

The adapted Bush-Monsteller algorithm is based on a stochastic approximation approach. Nodes should update their strategies following a probability distribution in order to learn. In the algorithm we propose, probabilities are updated  with respect to the response of the environment. We denote by~$p_{{ik}}$ the probability that MS $i$ chooses the power level $k$ (we suppose a discrete power space with a given fixed number of power levels for each user $\mathcal{P}^i=\{P_i^1,\dots,P_i^{max}\}$). Power strategies and utilities are suffixed by~$t$ to refer to their values at time step $t$. The probabilities are calculated as follows

\begin{equation}\label{proba}p_{{bk}}^{{t+1}} = \left\{ 
\begin{array}{l l}
  p_{{ik}}^{{t}}-\zeta \frac{\left|\theta_i-r_{i}^{t}(\textbf{P}^t)\right|}{\max_{{t'}\leq {t}}(\left|\theta_i-r_{i}(\textbf{P}^{t'})\right|)} p_{{ik}}^{{t}} &  \text{if $P_{i} \neq P_i^{k}$}\\
  p_{{ik}}^{{t}}+\zeta \frac{\left|\theta_i-r_{i}^{t}(\textbf{P}^t)\right|}{\max_{{t'}\leq {t}}(\left|\theta_i-r_{i}(\textbf{P}^{t'})\right|)}\sum \limits_{{k'}\neq {k}}p_{{ik'}}^{{t}} &  \text{otherwise}\\ \end{array} \right., \end{equation}\par
\noindent where
%
$\zeta$ is a parameter in $\left[0,1\right]$ also called \textit{the step size of the probability updating rule}. The results of the algorithm are more accurate when $\epsilon$ is close to zero~\cite{Bush}. $\zeta$ can be different from a player to another in order to allow heterogeneous learning speed.

Hence, adapted Bush-Monsteller iterates as described in~\ref{alg:FLAPH}:
\begin{algorithm}[H]
  \caption{ELRI algorithm
    \label{alg:FLAPH}}
     \begin{algorithmic}[1] 
					\Statex{\textbf{Parameters initialization}}
					\State{Each user picks randomly a power level with initial probabilities vector $p_{ik}^{0}$}
					\Repeat
					\Statex{\textbf{Learning pattern}}
      \For{Each user $i$}
			\For{Each power state $k$}
			\If{$P_{i} \neq P^{k} $}
			     \State{$p_{ik}^{t+1} \gets  p_{{ik}}^{{t}}-\zeta \frac{\left|\theta_i-r_{i}^{t}(\textbf{P}^t)\right|}{m_t}$ }
				   \Else
				   \State{$p_{ik}^{t+1} \gets   p_{{ik}}^{{t}}+\zeta \frac{\left|\theta_i-r_{i}^{t}(\textbf{P}^t)\right|}{m_t}\sum \limits_{{k'}\neq {k}}p_{{ik'}}^{{t}}$ }
				\EndIf
			     \State{Update $m_t$}
					\If{$\left|\theta_i-r_{i}^{t}(\textbf{P}^t)\right|>m_t$}
					\State{$m_{t+1}\gets \left|\theta_i-r_{i}^{t}(\textbf{P}^t)\right|$}
					\EndIf
				   \EndFor
					\EndFor
					\Until{The stopping criterion}
  \end{algorithmic}
\end{algorithm}

It is to be noted that the adapted Bush Monsteller will converge to the ESE only if it starts near from the ESE and does not meet an NE on its trajectory, otherwise it would be trapped in an NE.

\subsection{Learning for random channels}

When the channels are random, particularly when they are fast varying over time, we would like to track the efficient LTSE online. In contrast to the stationary case, the learning process cannot be based on the instantaneous throughput only. Players should estimate their throughput at each iteration and adjust their estimated throughput accordingly. Furthermore, when the channel fading is fast, the estimated throughput is not a contractive function since it depends on both powers and channel states. Consequently, the classical Banach-Picard updating rule cannot ensure convergence to the fixed point. To deal with this matter, Mann iterates are proposed~\cite{Mann1953}. Mann iterates modify Banach-Picard updating rule by taking an average on the previous power state and the Banach-Picard rule. This learning process, as described in~\ref{alg:Mann}, is proved to converge to the fixed point for functions that are not necessarily contractive. 

\begin{algorithm}
  \caption{Mann iterates for random channels
    }\label{alg:Mann}
     \begin{algorithmic}[1] 
					\State{\textbf{Parameters initialization}}
					\Statex{Each user picks randomly a transmit power with initial probabilities vector $P_{i}^{0}$}
					\Statex{$\hat{r}_i^{0}(P_i^t,\textbf{P}_{-i})$ is initiated}
					\Repeat
					\State{\textbf{Learning pattern}}
      \For{Each user $i$}
			   \Statex{Update its power as follows:}
			   \Statex{$P_{i}^{t+1} \gets (1-\lambda) P_i^t+\lambda P_{i}^t\frac{\theta_{i}}{\hat{r}_i^t}$ }
							   \Statex{Forecast the expected value of its throughput ${\hat{r}_i^{t+1}\!(P_i^t,\textbf{P}_{-i})\!=\!\hat{r}_i^{t}(P_i^t,\textbf{P}_{-i})\!+\!\mu(\hat{r}_i^{t}(P_i^t,\textbf{P}_{-i})\!-\!{r}_i^{t}(P_i^t,\textbf{P}_{-i})\!)}$ }


			\EndFor
					\Until{The stopping criterion}
  \end{algorithmic}
\end{algorithm}
with $(\lambda,\mu) \in [0,1]^2$ enough small in order to achieve convergence. Notice that when $\lambda =1$, the Mann iterates coincide with Banach-Picard algorithm. 
\section{Simulation Results}\label{Simulations}
Here, we turn to present some representative simulations by which we validate our analysis and show the performance of the proposed algorithms.
\subsection{QoS requirement and channel gain effects}
First, we suppose stationary channels and continuous power states spaces. We use Banach-Picard algorithm to reach the ESE. Consider Fig.~\ref{QoS} which depicts the transmit power evolution and subsequent throughput for each user. As the channel gain is supposed fixed ($h_i=1$ for all $i$), final power decisions (Fig.~\ref{QoS}~(a)) follow the QoS requirements. For example, user~$3$, who is the most demanding among users, transmits with the highest power. This remark is supported by Proposition~\ref{order} which states that the received powers order follows the users requirements order. It can also be seen from Fig.~\ref{QoS}~(b) that the final throughput of each user matches perfectly its initial demand. Besides, the ESE is met in only a few ten iterations. Notice that the some of transmit powers is set to the minimum at the end of the iterations which goes in the same direction with our prediction: the ESE minimizes the sum of transmit powers among all the satisfaction equilibria.  

\begin{figure}
\begin{tabular}{c}
\psfrag{Transmit power}[c][c][0.7]{Transmit power}
\psfrag{Elapsed time}[c][c][0.7]{Elapsed time}
\includegraphics[width=8cm,height=4cm]{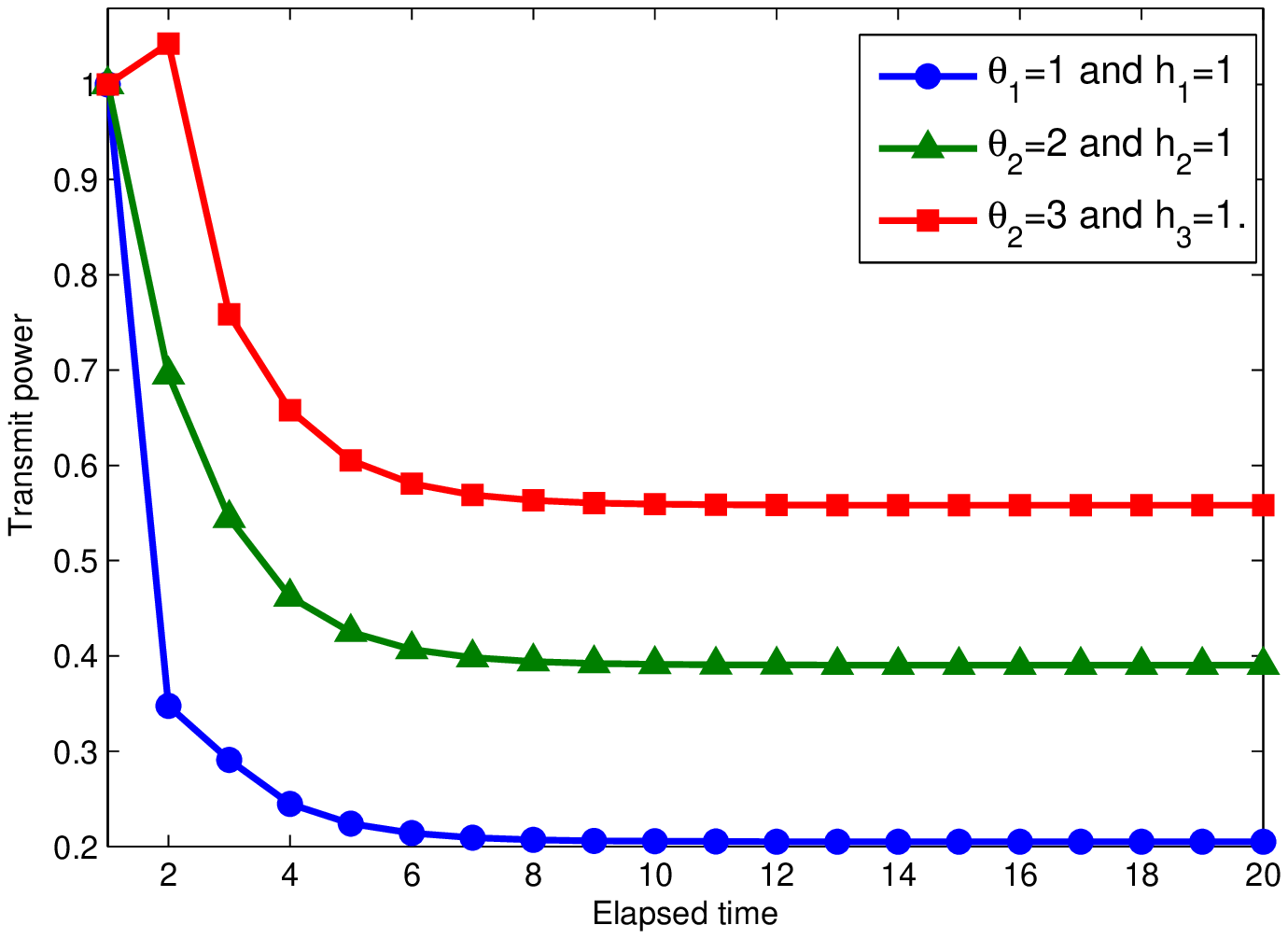}

\\
(a)
\\
\psfrag{Average throughput}[c][c][0.7]{Instantaneous throughput}
\psfrag{Elapsed time}[c][c][0.7]{Elapsed time}
\includegraphics[width=8cm,height=4cm]{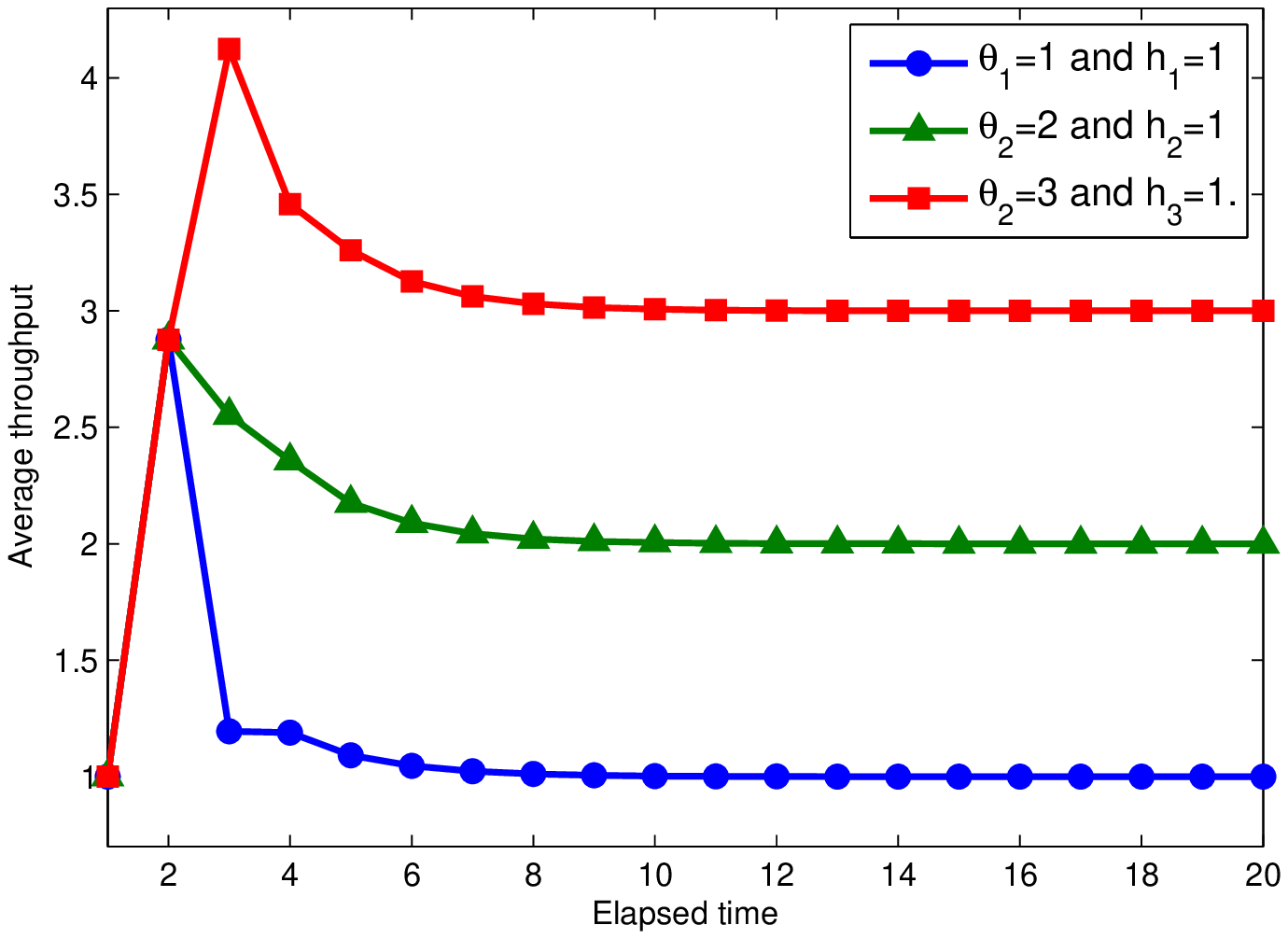}
\\
(b)
\end{tabular}
\caption{QoS requirement effect. Figure (a) represents the transmit powers decision evolution over time for $3$ users. The resulting throughput is given by Figure (b). The starting power allocation is $(1,1,1)[mW]$, and the maximum network capacity $C=10$.}
\label{QoS}
\end{figure}
Similarly, in Fig.~\ref{Channell}, we suppose that the QoS requirements (users demands) are fixed, and notice that the transmit powers order is inverted when compared to the channel gains order. This is explained by Proposition~\ref{order}, and is also intuitively expected: the noisier the channel state is, the more transmit power is needed to reach the satisfaction levels. The figure shows also that the users requirements are exactly met by the end of the iterations, and the energy cost is optimized.  
\begin{figure}
\begin{tabular}{c}
\psfrag{Transmit power}[c][c][0.7]{Transmit power}
\psfrag{Elapsed time}[c][c][0.7]{Elapsed time}
\includegraphics[width=8cm,height=4cm]{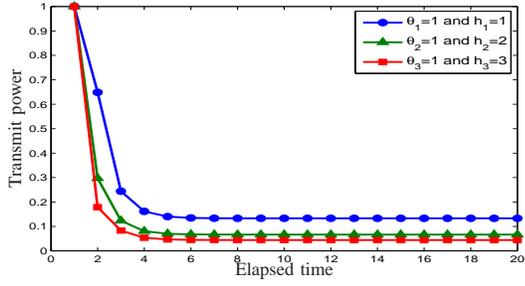}
\\
(a)\\
\psfrag{Average throughput}[c][c][0.7]{Instantaneous throughput}
\psfrag{Elapsed time}[c][c][0.7]{Elapsed time}

\includegraphics[width=8cm,height=4cm]{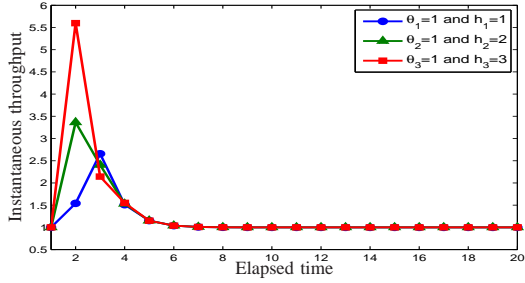}
\\
(b)\\
\end{tabular}
\caption{Channel gain effect. Figure (a) represents the transmit powers decision evolution over time for $3$ users. The resulting throughput is given by Figure (b). The starting power allocation is $(1,1,1)[mW]$, and the maximum network capacity $C=10$.}
\label{Channell}
\end{figure}

\begin{figure}
\begin{tabular}{c}
\psfrag{Transmit power}[c][c][0.7]{Transmit power}
\psfrag{Elapsed time}[c][c][0.7]{Elapsed time}
\includegraphics[width=8cm,height=4cm]{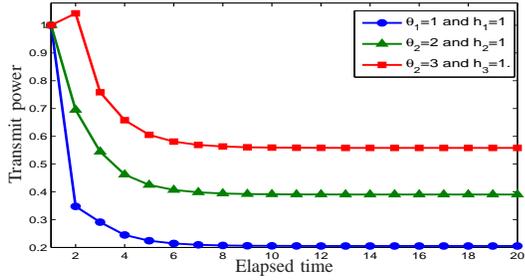}

\\
(a)
\\
\psfrag{Average throughput}[c][c][0.7]{Instantaneous throughput}
\psfrag{Elapsed time}[c][c][0.7]{Elapsed time}
\includegraphics[width=8cm,height=4cm]{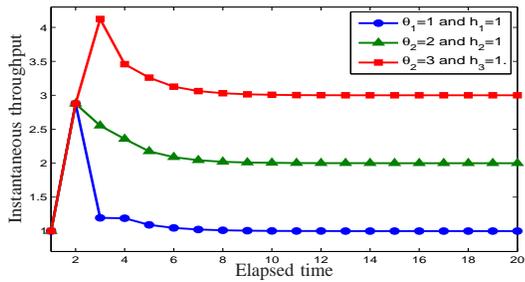}
\\
(b)
\end{tabular}
\caption{QoS requirement effect. Figure (a) represents the transmit powers decision evolution over time for $3$ users. The resulting throughput is given by Figure (b). The starting power allocation is $(1,1,1)[mW]$, and the maximum network capacity $C=10$.}
\label{QoS}
\end{figure}
\subsection{Adapted Bush-Monsteller algorithm}
When the channel states are stationary and the power levels are discrete, one can use the adapted Bush-Monsteller algorithm in order to reach the ESE.
Consider Fig.~\ref{Channel22} which depicts the algorithm convergence for $3$ users. From this figure, it can be seen that each player reaches its throughput requirement after a given number of iterations. It is to be noted that as the power states space is discrete, powers that exactly match the users requirements may not exist, and thus, users requirements are reached approximately. Mixed strategies can also be tracked by an adjustment in the learning pattern of the proposed algorithm using the Linear Reward Penalty algorithm~\cite{scha}.

\begin{figure}[t]
\begin{tabular}{c}
\psfrag{Instantaneous throughput}[c][c][0.7]{Instantaneous throughput}
\psfrag{Elapsed time}[c][c][0.7]{Elapsed time}
\includegraphics[width=8cm,height=4cm]{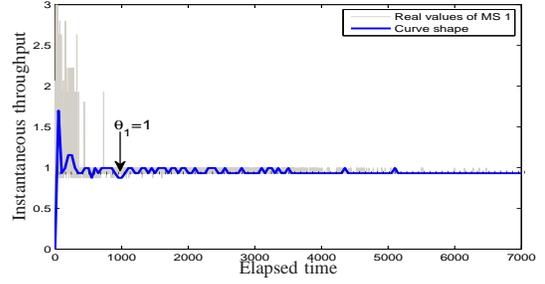}
\\
(a) \small{Throughput of user $1$} \\
\psfrag{Instantaneous throughput}[c][c][0.7]{Instantaneous throughput}
\psfrag{Elapsed time}[c][c][0.7]{Elapsed time}

\includegraphics[width=8cm,height=4cm]{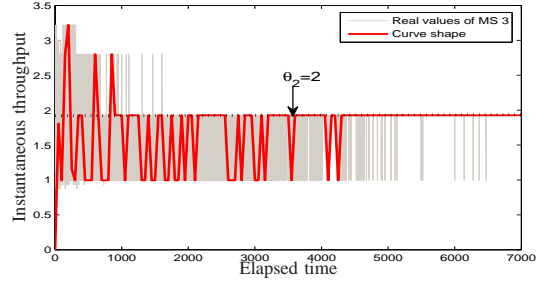}
\\
(b)\small{Throughput of user $2$}\\

\\
\psfrag{Instantaneous throughput}[c][c][0.7]{Instantaneous throughput}
\psfrag{Elapsed time}[c][c][0.7]{Elapsed time}
\includegraphics[width=8cm,height=4cm]{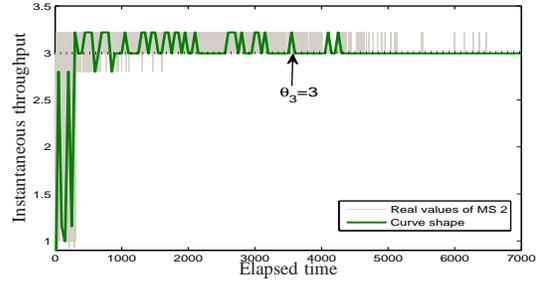}
\\
(c) \small{Throughput of user $3$}\\
\end{tabular}
\caption{Throughput convergence to the target demand while using adapted Bush-Monsteller algorithm. The starting power allocation is $(1,1,1)[mW]$, and the maximum network capacity $C=10$. The power states space are the same for all users, mainly $\{0.1,0.2,0.3\}$.The final power allocation is $(0.1,0.2,0.3)[mW]$}
\label{Channel22}
\end{figure}
Fig.\ref{Channel25} describes the probabilities of choosing a given power level for MS $1$. Clearly, as final power decision of the user is $P_1=0.1$, the probability to choose $P_1=0.1$ converges to $1$, and the other probabilities that are related to the other power states fail to $0$.

\begin{figure}[t]
\psfrag{Probability}[c][c][0.7]{Probability}
\psfrag{Elapsed time}[c][c][0.7]{Elapsed time}

\includegraphics[width=8cm,height=4cm]{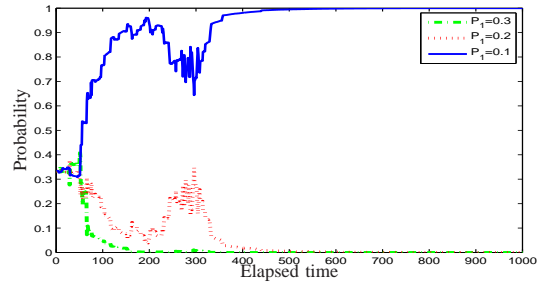}
\caption{The adapted Bush-Monsteller algorithm probabilities related to user $1$ with respect to its possible power states $\{0.1,0.2,0.3\}$. The starting power allocation is $(1,1,1)[mW]$, and the maximum network capacity $C=10$.}
\label{Channel25}

\end{figure}

\subsection{Random channel}
\subsubsection{Stationary block channels}
In this subsection, we suppose that the channels remain stationary over a given coherence time (a block) and change from one block to another. We use Banach-Picard iterations to select among continuous sets of powers the ESE. Here we highlight the tractable property of Banach-Picard algorithm. For each time block, the users requirements are reached. The channels states vary within iteration blocks as it can be seen in Fig.~\ref{Channel}.(b). The ESE is tracked for each time block. Particularly, in Fig.~\ref{Channel}.(a) the powers are adjusted (lowered for the first time block of 10 ms) in order to meet the users demands in Fig.~\ref{Channel}.(c). However, these results can not hold true when the time block is shorter than the convergence time of Banach-Picard algorithm.    

\begin{figure}
\begin{tabular}{c}
\psfrag{Transmit power}[c][c][0.7]{Transmit power}
\psfrag{Elapsed time}[c][c][0.7]{Elapsed time}
\includegraphics[width=8cm,height=4cm]{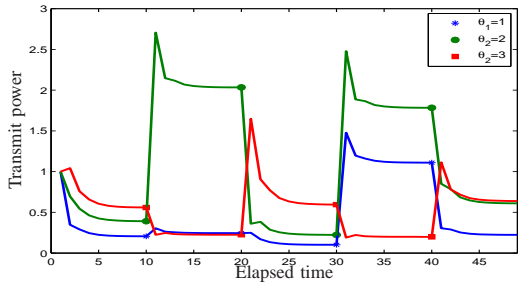}
\\
(a)\\
\psfrag{Channel gain}[c][c][0.7]{Channel gain}
\psfrag{Elapsed time}[c][c][0.7]{Elapsed time}

\includegraphics[width=8cm,height=4cm]{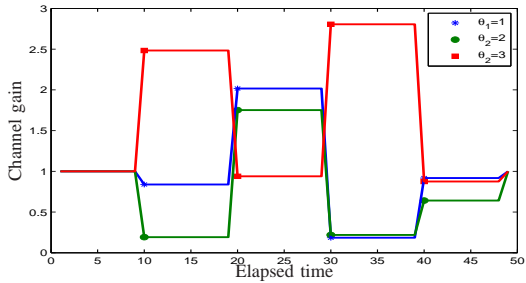}
\\
(b)\\

\\
\psfrag{Average throughput}[c][c][0.7]{Instantaneous throughput}
\psfrag{Elapsed time}[c][c][0.7]{Elapsed time}

\includegraphics[width=8cm,height=4cm]{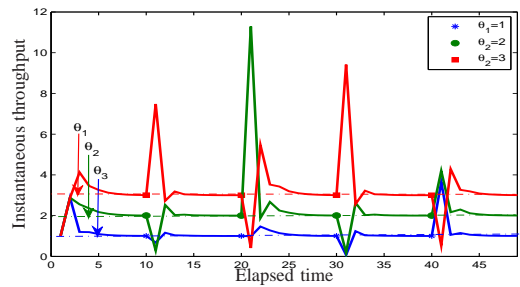}
\\
(c)\\
\end{tabular}
\caption{Stationary block channels. Channels vary every $10$ iterations. Figure (a) represents the transmit powers decision evolution over time for $3$ users. The resulting throughput is given by Figure (b). The starting power allocation is $(1,1,1)[mW]$, and the maximum network capacity $C=10$.}
\label{Channel}
\end{figure}
\subsubsection{Fast fading channels}
In order to cope with fast varying channels, Mann iterates algorithm is used. Fig.~\ref{ChannelRand}.(c) shows that the algorithm reaches the efficient LTSE in long run iterations. The variations around the users requirements are barely canceled at the end of the iterations. Notice that channels are completely random as depicted in Fig.~\ref{ChannelRand}.(b) following an exponential distribution. The updated transmit powers are stabilized at the end of the running time as it can be seen in Fig.~\ref{ChannelRand}.(a).
\begin{figure}
\begin{tabular}{c}
\psfrag{Transmit power}[c][c][0.7]{Transmit power}
\psfrag{Elapsed time}[c][c][0.7]{Elapsed time}
\includegraphics[width=8cm,height=4cm]{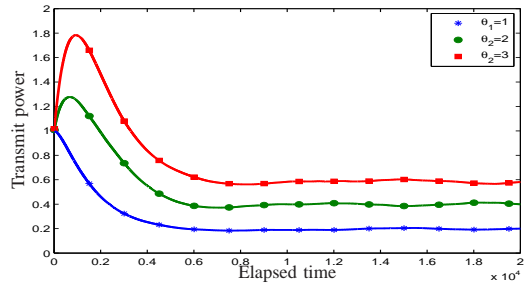}
\\
(a)\\
\psfrag{Channel gain}[c][c][0.7]{Channel gain}
\psfrag{Elapsed time}[c][c][0.7]{Elapsed time}

\includegraphics[width=8cm,height=4cm]{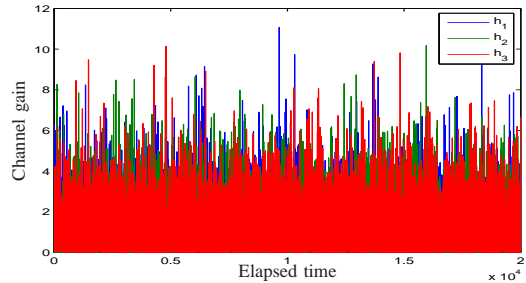}
\\
(b)\\

\\
\psfrag{Average throughput}[c][c][0.7]{Expected throughput}
\psfrag{Elapsed time}[c][c][0.7]{Elapsed time}

\includegraphics[width=8cm,height=4cm]{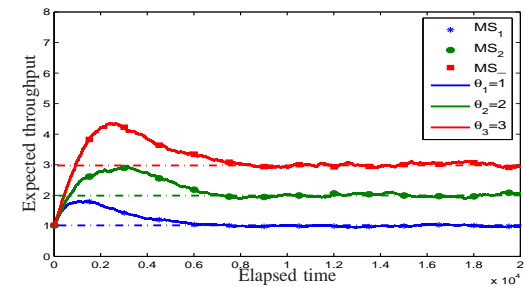}
\\
(c)\\
\end{tabular}
\caption{Random channel gains follow an exponential distribution of mean and variance equal to $1$. Figure (a) represents the transmit powers decision evolution over time for $3$ users. The resulting throughput is given by Figure (c). Figure (c) gives the channels variation over time. The starting power allocation is $(1,1,1)[mW]$, and the maximum network capacity $C=10$.}
\label{ChannelRand}
\end{figure}

%
%
%


\subsection{Algorithm recovery}
Here, we turn to show the recovery property of Banach-Picard iterations. By the recovery property, we mean the algorithm ability to adapt its results when the environment parameters change. Indeed, if a user changes its demand over time a lower number of iterations should be needed in order to return back to the ESE allocation. Fig.~\ref{Heal} shows that user $1$ increases its demand at $t=20$. This demand is quickly satisfied and the algorithm converges rapidly, again, to the new ESE.

\begin{figure}
\begin{tabular}{c}
\psfrag{Transmit power}[c][c][0.7]{Transmit power}
\psfrag{Elapsed time}[c][c][0.7]{Elapsed time}
\includegraphics[width=8cm,height=4cm]{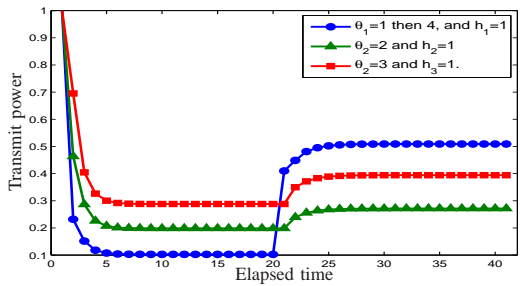}\\
(a)\\
\psfrag{Average throughput}[c][c][0.7]{Average throughput}
\psfrag{Elapsed time}[c][c][0.7]{Elapsed time}
\includegraphics[width=8cm,height=4cm]{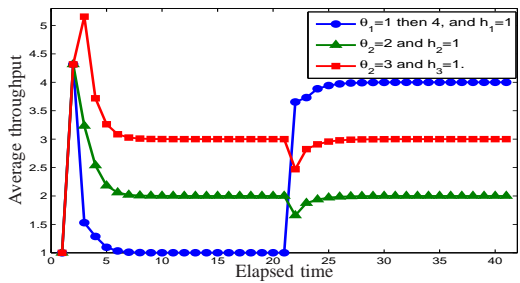}\\
(b)\\
\end{tabular}
\caption{Algorithm recovery. Figure (a) represents the transmit powers decision evolution over time for $3$ users. The resulting throughput is given by Figure (b). The first user increases its demand at $t=20$. The starting power allocation is $(1,1,1)[mW]$, and the maximum network capacity $C=10$.}
\label{Heal}
\end{figure}


\subsection{Capacity discovery}

Here, we are interested in maximizing the payoffs using the progressive Banach-Picard algorithm for capacity discovery.
Fig.~\ref{CNE} plots the transmit powers and payoffs evolution within time. It can be seen from the figure that users reach their optimal payoffs by the end of the iterations. It is worth noting that the sum of their final requirements corresponds to the maximum network capacity. 
\begin{figure}[t]
\begin{tabular}{c}
\psfrag{Transmit power}[c][c][0.7]{Transmit power}
\psfrag{Elapsed time}[c][c][0.7]{Elapsed time}
\includegraphics[width=8cm,height=3.8cm]{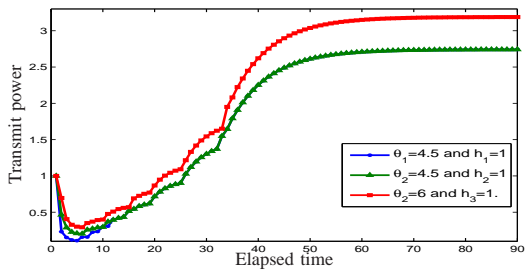}\\
(a)\\
\psfrag{Average throughput}[c][c][0.7]{Average throughput}
\psfrag{Elapsed time}[c][c][0.7]{Elapsed time}
\includegraphics[width=8cm,height=3.8cm]{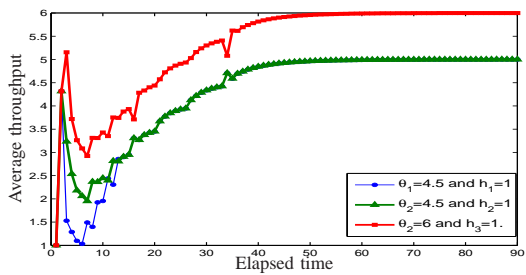}\\
(b)\\
\end{tabular}
\caption{Throughput maximization using the modified Banach-Picard algorithm. Figure (a) represents the transmit powers decision evolution over time for $3$ users. The resulting throughput is given by Figure (b). The starting power allocation is $(1,1,1)[mW]$, and the maximum network capacity $C=14$.}
\label{CNE}
\end{figure}
\subsection{Constraints feasibility}
In this section, we suppose, that the sum of initial users requirements exceeds the network capacity. Fig.~\ref{Band} depicts the users behaviors when their requirements ($\sum \limits_{i=1}^3 \theta_i=6$) go beyond the network capacity ($C=3$). Hence, the users constraints are unfeasible (inequality~(\ref{Unfeasible}) is not satisfied). It can be seen from Fig.~\ref{Band}~(a) that users have to transmit with an infinite power in order to reach their satisfaction performance. Furthermore, even with the highest power levels, users are not able to achieve their request as it is depicted in Fig.~\ref{Band}~(b) where the final throughput of each user is below its satisfaction. The aggressive behavior of users, particularly transmitting with high powers, can be compared to the behavior of prisoners' dilemma players, in which the players behave in discordance to what is expected.

\begin{figure}[t]
\begin{tabular}{c}
\psfrag{Transmit power}[c][c][0.7]{Transmit power}
\psfrag{Elapsed time}[c][c][0.7]{Elapsed time}

\includegraphics[width=8cm,height=4cm]{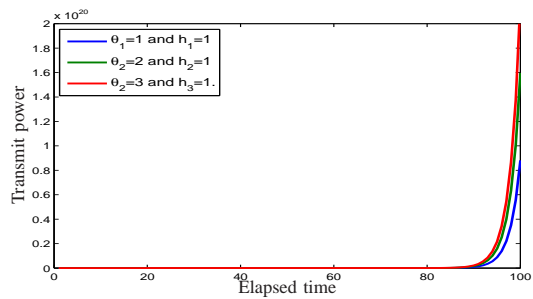}
\\
(a)
\\
\psfrag{Average throughput}[c][c][0.7]{Average throughput}
\psfrag{Elapsed time}[c][c][0.7]{Elapsed time}
\includegraphics[width=8cm,height=4cm]{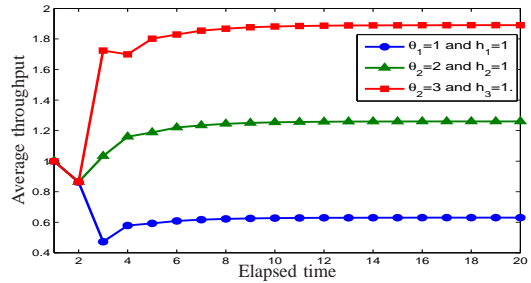}
\\
(b)
\\
\end{tabular}
\caption{Constraints feasibility. Figure (a) represents the transmit powers decision evolution over time for $3$ users. The resulting throughput is given by Figure (b). The starting power allocation is $(1,1,1)[mW]$, and the maximum network capacity $C=3$.}
\label{Band}
\end{figure}
\section{Conclusion}\label{conclu}
In this paper, we proposed a satisfaction equilibrium approach in order to reduce energy consumption for users in self-organizing networks. We were mainly interested in the efficient satisfaction equilibrium for both stationary and fast fading channels. We showed that we studied conditions
of existence and uniqueness of ESE under a stationary channel
assumption and fully characterized the ESE by the solution of a linear system. Moreover, we proved that at 
the ESE no player can increase its QoS without degrading the overall
energy performance. When the channels are fast varying over time the efficient long term satisfaction equilibrium solution is proposed and characterized. In order to reach to the efficient solutions, four fully distributed algorithms are proposed, namely: Banach-Picard algorithm, the progressive Banach-Picard algorithm for capacity discovery, the efficient linear reward algorithm, and Mann iterates algorithm. The convergence of the proposed algorithms is shown through simulation results and their qualitative properties are described are well. 

\appendices 
\section{Proof of the Proposition 4:}\label{Appen}
\begin{proof}
 The ESE is the solution of the linear system of equations:$\forall i \in \mathcal{U}$
\begin{equation}\label{throu}
r_i(\textbf{P})=\theta_i.
\end{equation}.
The solution exists if
\begin{itemize}
	\item The determinant is different from zero.
	\item The power solutions are positive.
	\item The power solutions do not exceed the maximum transmit power allowed to each user. 
\end{itemize}
First, we compute the determinant of the linear system. We replace equation~(\ref{th}) in equation~(\ref{throu}), after a few computations, we obtain, $\forall i \in \mathcal{U}$ 
\begin{equation}
P_ih_i-(2^{\theta_i}-1)\sum \limits_{j\neq i} P_jh_j=(2^{\theta_i}-1)\eta.
\end{equation}
Let $A$ be an $N\times N$ matrix, such that $A=(a_{ij})_{1\leq i,j\leq N}$, with $a_{ii}=h_i$ and for $i\neq j$ $a_{ij}=(1-2^{\theta_i})h_j$

\[A\!\!=\!\!\left(
\begin{matrix}

h_1 & (1\!\!-\!\!2^{\theta_1})h_2 & \dots& (1\!\!-\!\!2^{\theta_1})h_N\\
(1\!\!-\!\!2^{\theta_2})h_1 & h_2 &  \dots &(1\!\!-\!\!2^{\theta_2})h_N\\
\dots & \dots & h_i  & (1\!\!-\!\!2^{\theta_i})h_N\\
(1\!\!-\!\!2^{\theta_N})h_1 & (1\!\!-\!\!2^{\theta_N})h_2 & \dots& h_N\\

\end{matrix}\right).\]

Let $B=diag_{1 \leq i\leq N}(\frac{1}{1-2^{\theta_i}})$ 
\[B\!\!=\!\!\left(
\begin{matrix}

\frac{h_1}{(1-2^{\theta_1})} & h_2 & \dots& h_N\\
h_1 & \frac{h_2}{(1-2^{\theta_2})} &  \dots &h_N\\
\dots & \dots & \frac{h_i}{(1-2^{\theta_i})}  & h_N\\
h_1 & h_2 & \dots& \frac{h_N}{(1-2^{\theta_N})}\\

\end{matrix}\right)\]
The determinant of $B$ is given by
\begin{equation}
det(B)=\frac{1}{\prod \limits_{i=1}^{N}(1-2^{\theta_i})}det(A).
\end{equation}
Let $C=TB$ with $T=I_N-\sum \limits_{i=1}^{N-1}E_{i,i+1}$, where $I_N$ is the identity matrix, and $E_{i,i+1}$ is an elementary matrix where all components are equal to $0$ except the $\{i,i+1\}$th element.
\[C\!\!=\!\!\left(
\begin{matrix}

\frac{2^{\theta_1}h_1}{(1-2^{\theta_1})} & \frac{-2^{\theta_2}h_2}{(1-2^{\theta_2})} & 0&\dots& 0\\
0 & \frac{2^{\theta_2}h_2}{(1-2^{\theta_2})} & \frac{-2^{\theta_3}h_3}{(1-2^{\theta_3})}&  \dots &0\\
0 & 0 & \frac{2^{\theta_{i}}h_i}{(1-2^{\theta_i})} &\frac{-2^{\theta_{i+1}}h_{i+1}}{(1-2^{\theta_{i+1}})} &0\\
0 & \dots & 0 &\frac{2^{\theta_{i+1}}h_{i+1}}{(1-2^{\theta_{i+1}})} &0\\
h_1 & h_2 & \dots& \dots&\frac{h_N}{(1-2^{\theta_N})}\\

\end{matrix}\right).\]

%
%
%
Thus, we have

\begin{equation}
det(C)=det(B).
\end{equation}
From the computation of the determinant of $C$, it follows 
\begin{equation}\label{66}
\begin{aligned}
det(B)&=\left(\!\!\sum \limits_{i=1}^{N-1}\! h_i\!\!\prod \limits_{j=1,j\neq i}^{N} \!\!\frac{2^{\theta_j}h_j}{1-2^{\theta_j}}\!\!+\!\! \frac{h_N}{1-2^{\theta_N}}\!\! \prod \limits_{j=1}^{N-1} \frac{2^{\theta_j}h_j}{1-2^{\theta_j}}\!\!\right),\\
\end{aligned} 
\end{equation}

Finally,
\begin{equation}
\begin{aligned}
det(A)&=\left(\!\!\sum \limits_{i=1}^{N-1}\! h_i\!\!\prod \limits_{j=1,j\neq i}^{N} \!\!\frac{2^{\theta_j}h_j}{1-2^{\theta_j}}\!\!+\!\! \frac{h_N}{1-2^{\theta_N}}\!\! \prod \limits_{j=1}^{N-1} \frac{2^{\theta_j}h_j}{1-2^{\theta_j}}\!\!\right)\times\\
&\prod \limits_{i=1}^{N}(1-2^{\theta_i}) \\
& =\sum \limits_{i=1}^{N-1}\! h_i(1-2^{\theta_i})\!\!\prod \limits_{j=1,j\neq i}^{N} 2^{\theta_j}h_j\!\!+\!\! h_N\!\! \prod \limits_{j=1}^{N-1}2^{\theta_j}h_j\\
&\neq 0.
\end{aligned} 
\end{equation}

The solution of the linear system is given by (extensive computations are given in the appendix)

\begin{equation}
\Bigg\{
\begin{array}{l l}
P_N^+ &= \frac{\eta(2^{\theta_N}-1)}{\sum \limits_{i=1}^{N-1}h_i\frac{1-2^{\theta_N}}{2^{\theta_N}h_N}\frac{2^{\theta_i}h_i}{1-2^{\theta_i}}+\frac{h_N}{1-2^{\theta_N}}}\\
 P_i^+&=\frac{1-2^{\theta_i}}{2^{\theta_i}h_i}\frac{2^{\theta_N}h_N}{1-2^{\theta_N}}P_N.\\
\end{array}
\end{equation}
It is worth mentioning that the sign of $P_N^+$ is the same of $P_i^+$, the reason why we only condition $P_N^+$ sign. The solution of the system is accepted only if it fits into the power states set. 

We solve the following linear system:


\begin{equation}
\Bigg\{
\begin{array}{l l}
h_1P_1^+ + \dots+ (1\!\!-\!\!2^{\theta_1})h_N P_N^+ &=\eta( 2^{\theta_1}-1)\\
(1\!\!-\!\!2^{\theta_2})h_1 P_1^++  \dots +(1\!\!-\!\!2^{\theta_2})h_N P_N^+&= \eta(2^{\theta_2}-1)\\
\dots & =\dots\\
(1\!\!-\!\!2^{\theta_N})h_1P_1^+  + \dots+ h_NP_N^+&=\eta(2^{\theta_N}-1)\\
\end{array}
\end{equation}
We divide each $i$th line  by $\eta( 2^{\theta_i}-1)$,
\begin{equation}
\Bigg\{
\begin{array}{l l}
\frac{h_1}{(1-2^{\theta_1})}P_1^+ + h_2P_2^+ + \dots+ h_NP_N^+&=-1\\
h_1P_1^+ + \frac{h_2}{(1-2^{\theta_2})}P_2^+ +  \dots +h_NP_N^+&=-1\\
\dots & =\dots\\
h_1P_1^+ + h_2P_2^+ + \dots+ \frac{h_N}{(1-2^{\theta_N})}P_N^+&=\eta(2^{\theta_N}-1).\\
\end{array}
\end{equation}
For the first $N-1$ lines, we deduce the $i+1$th line from the $i^{th}$ line, and we obtain

\begin{equation}
\Bigg\{
\begin{array}{l l}

\frac{2^{\theta_1}h_1}{(1-2^{\theta_1})}P_1^+ + \frac{-2^{\theta_2}h_2}{(1-2^{\theta_2})}P_2^+ &= 0\\
 \frac{2^{\theta_2}h_2}{(1-2^{\theta_2})}P_2^+ + \frac{-2^{\theta_3}h_3}{(1-2^{\theta_3})}P_3^+&= 0\\
 \frac{2^{\theta_{i}}h_i}{(1-2^{\theta_i})}P_i^+ +\frac{-2^{\theta_{i+1}}h_{i+1}}{(1-2^{\theta_{i+1}})}P_{i+1}^+ &=0\\
\dots&=\dots\\
h_1P_1^++h_2P_2^++ \dots+\frac{h_N}{(1-2^{\theta_N})}P_N^+&=\eta(2^{\theta_N}-1).\\

\end{array}
\end{equation}

We replace each $i$th line by the sum from the first to the $i$th line, 
\begin{equation}
\Bigg\{
\begin{array}{l l}

\frac{2^{\theta_1}h_1}{(1-2^{\theta_1})}P_1^+ + \frac{-2^{\theta_N}h_N}{(1-2^{\theta_N})}P_N^+ &= 0\\
 \frac{2^{\theta_2}h_2}{(1-2^{\theta_2})}P_2^+ + \frac{-2^{\theta_N}h_N}{(1-2^{\theta_N})}P_N^+&= 0\\
 \frac{2^{\theta_{i}}h_i}{(1-2^{\theta_i})}P_i^+ +\frac{-2^{\theta_N}h_N}{(1-2^{\theta_N})}P_N^+ &=0\\
\dots&=\dots\\
h_1P_1+h_2P_2+ \dots+\frac{h_N}{(1-2^{\theta_N})}P_N^+&=\eta(2^{\theta_N}-1)\\

\end{array}
\end{equation}

By replacing in the last line, we find $P^+_N$
\begin{equation}
\Bigg\{
\begin{array}{l l}
\frac{2^{\theta_1}h_1}{(1-2^{\theta_1})}P_1^+ + \frac{-2^{\theta_N}h_N}{(1-2^{\theta_N})}P_N^+ &= 0\\
 \frac{2^{\theta_2}h_2}{(1-2^{\theta_2})}P_2^+ + \frac{-2^{\theta_N}h_N}{(1-2^{\theta_N})}P_N^+&= 0\\
 \frac{2^{\theta_{i}}h_i}{(1-2^{\theta_i})}P_i^+ +\frac{-2^{\theta_N}h_N}{(1-2^{\theta_N})}P_N^+ &=0\\
\dots&=\dots\\
(\sum \limits_{i=1}^{N-1}h_i\frac{1-2^{\theta_N}}{2^{\theta_N}h_N}\frac{2^{\theta_i}h_i}{1-2^{\theta_i}}+\frac{h_N}{1-2^{\theta_N}})P_N^+&=\eta(2^{\theta_N}-1).\\
\end{array}
\end{equation}
Finally, we obtain,

\begin{equation}
\Bigg\{
\begin{array}{l l}
P_N^+ &= \frac{\eta(2^{\theta_N}-1)}{\sum \limits_{i=1}^{N-1}h_i\frac{1-2^{\theta_N}}{2^{\theta_N}h_N}\frac{2^{\theta_i}h_i}{1-2^{\theta_i}}+\frac{h_N}{1-2^{\theta_N}}}\\
 P_i^+&=\frac{1-2^{\theta_i}}{2^{\theta_i}h_i}\frac{2^{\theta_N}h_N}{1-2^{\theta_N}}P_N^+.\\
\end{array}
\end{equation}
Which completes our proof.
\end{proof}

\balance


\end{document}